\documentclass[12pt]{article}

\pdfoutput=1

\usepackage[hmargin=0.98in,vmargin=0.9in]{geometry}

\usepackage[
  bookmarks=true,
  bookmarksnumbered=true,
  bookmarksopen=true,
  pdfborder={0 0 0},
  breaklinks=true,
  colorlinks=true,
  linkcolor=black,
  citecolor=black,
  filecolor=black,
  urlcolor=black,
]{hyperref}
\usepackage[round]{natbib}
\usepackage{amssymb,amsmath,amsthm}
\usepackage[capitalise]{cleveref}
\usepackage{nicefrac}
\usepackage{tikz}
\usetikzlibrary{patterns}
\usepackage{pgfplots}
\pgfplotsset{compat=1.14}

\usepackage{fnbreak}

\allowdisplaybreaks

\theoremstyle{plain}
\newtheorem{theorem}{Theorem}

\newtheorem{proposition}[theorem]{Proposition}
\newtheorem{corollary}[theorem]{Corollary}
\Crefname{equation}{Equation}{Equations} 
\Crefname{figure}{Figure}{Figures} 

\theoremstyle{definition}
\newtheorem{definition}{Definition}

\newcommand{\eqdef}{\triangleq}

\DeclareMathOperator{\OPT}{OPT}
\DeclareMathOperator{\Rev}{Rev}

\newcommand{\deltafloor}[1]{\lfloor#1\rfloor_{\delta}}

\hypersetup{
  pdfauthor      = {Yannai A. Gonczarowski <yannai@gonch.name>},
  pdftitle       = {Bounding the Menu-Size of Approximately Optimal Auctions via Optimal-Transport Duality},
}

\title{\texorpdfstring{\vspace{-2em}}{}Bounding the Menu-Size of\texorpdfstring{\\}{ }Approximately Optimal Auctions\texorpdfstring{\\}{ }via Optimal-Transport Duality}
\author{Yannai A. Gonczarowski\thanks{Einstein Institute of Mathematics, Rachel \& Selim Benin School of Computer Science \& Engineering, and Federmann Center for the Study of Rationality, The Hebrew University of Jerusalem, Israel; and Microsoft Research. \emph{E-mail}: \href{mailto:yannai@gonch.name}{yannai@gonch.name}.}}
\date{July 11, 2018}

\begin{document}
\maketitle

\begin{abstract}
\looseness=-1
The question of the minimum menu-size for approximate (i.e., up-to-$\varepsilon$) Bayesian revenue maximization when selling two goods to an additive risk-neutral quasilinear buyer was introduced by \citet{hn13}, who give an upper bound of $O(\nicefrac{1}{\varepsilon^4})$ for this problem. Using the optimal-transport duality framework of \citet{ddt13,ddt15}, we derive the first lower bound for this problem --- of $\Omega(\nicefrac{1}{\sqrt[4]{\varepsilon}})$, even when the values for the two goods are drawn i.i.d.\ from ``nice'' distributions, establishing how to reason about approximately optimal mechanisms via this duality framework. This bound implies, for any fixed number of goods, a tight bound of $\Theta(\log\nicefrac{1}{\varepsilon})$ on the minimum deterministic communication complexity guaranteed to suffice for running some approximately revenue-maximizing mechanism, thereby completely resolving this problem.
As a secondary result, we show that under standard economic assumptions on distributions, the above upper bound of \citet{hn13} can be strengthened to $O(\nicefrac{1}{\varepsilon^2})$.
\end{abstract}

\section{Introduction}

One of the high-level goals of the field of Algorithmic Mechanism Design is to understand the tradeoff between the economic efficiency and the simplicity of mechanisms, with a central example being auction mechanisms. One of the most fundamental scenarios studied in this context is that of revenue-maximization by a single seller who is offering for sale two or more goods to a single buyer. Indeed, while classic economic analysis \citep{m81} shows that for a single good, the revenue-maximizing mechanism is extremely simple to describe, it is known that the optimal auction for even two goods may be surprisingly complex and unintuitive \citep{mm88,t04,mv06,hr15,ddt13,ddt15,gk14,gk15}, eluding a general description to date.

In this paper we study, for a fixed number of goods, the tradeoff between the complexity of an auction and the extent to which it can approximate the optimal revenue.
While one may choose various measures of auction complexity \citep{hn13,dhn14,mr15}, we join several recent papers by focusing on the simplest measure, the \emph{menu-size} suggested by \citet{hn13}. While previous lower bounds on the menu-size as a function of the desired approximation to the revenue all assume a coupling between the number of goods and the desired approximation (so that the former tends to infinity simultaneously with the latter tending to optimal; e.g., setting~${\varepsilon\eqdef\nicefrac{1}{n}}$), in this paper we focus on the behavior of the menu-size as a function only of the desired approximation to the revenue, keeping the number of goods fixed and uncoupled from it. In particular, we obtain the first lower-bound on the menu-size that is not asymptotic in the number of goods, thereby quantifying the degree to which the menu-size of an auction really is a bottleneck to extracting high revenue even for a fixed number of goods.

\vspace{-.5em}\paragraph{Revenue Maximization}\looseness=-1 We consider the following classic setting. A risk-neutral seller has two goods for sale. A risk-neutral quasilinear buyer has a valuation (maximum willingness to pay) $v_i\in[0,1]$ for each of these goods~$i$,\footnote{Having $1$ rather than any other upper bound is without loss of generality, as the units are arbitrary.} and has an additive valuation (i.e., values the bundle of both goods by the sum of the valuations for each good). The seller has no access to the valuations~$v_i$, but only to a joint distribution $F$ from which they are drawn.
The seller wishes to devise a (truthful) auction mechanism for selling these goods, which will maximize her revenue among all such mechanisms, in expectation over the distribution~$F$. (The seller has no use for any unsold good.) We denote the maximum obtainable expected revenue by~$\OPT(F)$.

\vspace{-.5em}\paragraph{Menu-Size}
\citet{hn13} have introduced the \emph{menu-size} of a mechanism as a measure of its complexity:
this measure counts the number of possible outcomes of the mechanism (where an outcome is a specification of an allocation probability for each good, coupled with a price).\footnote{By the Taxation Principle, any mechanism is essentially described by the menu of its possible outcomes, as the mechanism amounts to the buyer choosing from this menu an outcome that maximizes her utility.} \citet{ddt13} have shown that even in the case of independently distributed valuations for the two goods, precise revenue maximization may require an infinite menu-size; \citet{ddt15} have shown this even when the valuations for the two goods are drawn i.i.d.\ from ``nice'' distributions. In light of these results, relaxations of this problem, allowing for mechanisms that maximize revenue up-to-$\varepsilon$, were considered.

\vspace{-.5em}\paragraph{Approximate Revenue Maximization}
\citet{hn13} have shown that a menu-size of~$O(\nicefrac{1}{\varepsilon^4})$ suffices for maximizing revenue up to an additive~$\varepsilon$:

\begin{theorem}[\citealp{hn13}]\label{hn}
There exists $C(\varepsilon)=O(\nicefrac{1}{\varepsilon^4})$ such that for every $\varepsilon>0$ and for every distribution $F\in\Delta\bigl([0,1]^2\bigr)$, there exists a mechanism $M$ with menu-size at most $C(\varepsilon)$ such that $\Rev_F(M)>\OPT(F)-\varepsilon$. ($\Rev_F(M)$ is the expected revenue of $M$ from~$F$.)
\end{theorem}

While the above-described results of \citeauthor{ddt15}\ imply that the menu-size required for up-to-$\varepsilon$ revenue maximization tends to infinity as $\varepsilon$ tends to $0$, no lower bound whatsoever was known on the speed at which it tends to infinity. (I.e., all that was known was that the menu-size is~$\omega(1)$ as a function of $\varepsilon$.) Our main result, which we prove in \cref{sec:lower}, is the first lower bound on the required menu-size for this problem,\footnote{In fact, we do not know of any previous menu-size bound, for any problem, that lower-bounds the menu-size as a function of $\varepsilon$ without having the number of goods $n$ also tend to infinity (e.g., by setting~$\varepsilon\eqdef\nicefrac{1}{n}$).} showing that a polynomial dependence on~$\varepsilon$ is not only sufficient, but also required, hence establishing that the menu-size really is a nontrivial bottleneck to extracting high revenue, even for two goods and even when the valuations for the goods are drawn i.i.d.\ from ``nice'' distributions:

\begin{theorem}[Menu-Size: Lower Bound]\label{lower}
There exist $C(\varepsilon)=\Omega(\nicefrac{1}{\sqrt[4]{\varepsilon}})$ and a distribution $F\in\Delta\bigl([0,1]\bigr)$, such that for every $\varepsilon>0$ it is the case that $\Rev_{F^2}(M)<\OPT(F^2)-\varepsilon$ for every mechanism~$M$ with menu-size at most $C(\varepsilon)$.
\end{theorem}

Our proof of \cref{lower} uses the optimal-transport duality framework of \citet{ddt13,ddt15}. To the best of our knowledge, this is the first use of this framework to reason about approximately optimal mechanisms, thereby establishing how to leverage this framework to do so.
 
\vspace{-.5em}\paragraph{Communication Complexity} As \citet{bgn17} show, the logarithm (base~$2$, rounded up) of the menu-size of a mechanism is precisely the deterministic communication complexity (between the seller and the buyer) of running this mechanism, when the description of the mechanism itself is common knowledge. Therefore, by \cref{hn,lower}, we obtain a tight bound on the minimum deterministic communication complexity guaranteed to suffice for running some up-to-$\varepsilon$ revenue-maximizing mechanism, thereby completely resolving this problem:

\begin{corollary}[Communuication Complexity: Tight Bound]\label{cc}
There exists $D(\varepsilon)=\Theta(\log\nicefrac{1}{\varepsilon})$ such that for every $\varepsilon>0$ it is the case that $D(\varepsilon$) is the minimum communication complexity that satisfies the following: for every distribution $F\in\Delta\bigl([0,1]^2\bigr)$ there exists a mechanism~$M$ such that the deterministic communication complexity of running~$M$ is $D(\varepsilon$) and such that $\Rev_F(M)>\OPT(F)-\varepsilon$. This continues to hold even if $F$ is guaranteed to be a product of two independent identical distributions.
\end{corollary}

In \cref{extensions}, we extend this tight communication-complexity bound to any fixed number of goods, as well as derive analogues of our results for multiplicative up-to-$\varepsilon$ approximation.

\vspace{1em}
\looseness=-1
While our lower bound completely resolves the open question of whether a polynomial menu-size is necessary (and not merely sufficient), and while it tightly characterizes the related communication complexity, it does not yet fully characterize the precise polynomial dependence of the menu-size on $\varepsilon$. While the proof of our lower bound (\cref{lower}) makes delicate use of a considerable amount of information regarding the optimal mechanism via the optimal-transport duality framework of \citet{ddt13,ddt15}, the proof of the upper bound of \citet{hn13} (\cref{hn}) makes use of very little information regarding the optimal mechanism. (As noted above, indeed very little is known regarding the structure of general optimal mechanisms.) Our secondary result, which we prove in \cref{sec:upper}, shows that under standard economic assumptions on valuation distributions, the upper bound of \citet{hn13} can be tightened by two orders of magnitude. This suggests that it may well be possible to use more information regarding the structure of optimal mechanisms, as such will be discovered, to unconditionally improve the upper bound of \citet{hn13}.

\begin{theorem}[Menu-Size: Conditional Upper Bound]\label{upper}
There exists $C(\varepsilon)=O(\nicefrac{1}{\varepsilon^2})$ such that for every $\varepsilon>0$ and for every distribution $F\in\Delta\bigl([0,1]^2\bigr)$ satisfying the McAfee-McMillan hazard condition (see \cref{hazard} in \cref{sec:upper}), there exists a mechanism $M$ with menu-size at most $C(\varepsilon)$ such that $\Rev_F(M)>\OPT(F)-\varepsilon$.
\end{theorem}

The question of a precise tight bound on the menu-size required for up-to-$\varepsilon$ revenue maximization remains open, for correlated as well as product (even i.i.d.)\ distributions.
We further discuss our results and their connection to other results and open problems in the literature in \cref{discussion}.

\section{Lower Bound}\label{sec:lower}

We prove our lower bound (i.e., \cref{lower}) by considering the setting in which \citet{ddt15} show that precise revenue maximization requires infinite menu-size, that is, the case of two items with i.i.d.\ valuations each drawn from the Beta distribution $F\eqdef B(1,2)$, i.e., the distribution over $[0,1]$ with density ${f(x)\eqdef2(1-x)}$. We first present, in \cref{sec:lower-overview}, a very high-level overview of the main proof idea in a way that does not go into any technical details regarding the duality framework of \citet{ddt13,ddt15}. We then present, in \cref{sec:lower-ddt}, only the minimal amount of detail from the extensive analysis of \citet{ddt15} that is required to follow our proof. Finally, in \cref{sec:lower-proof} we prove \cref{lower}.

\subsection{Proof Idea Overview}\label{sec:lower-overview}

We start by presenting a very high-level overview of the main idea of the proof of \cref{lower} in a way that does not go into any technical details regarding the duality framework of \citet{ddt13,ddt15}.

Fix a concrete distribution from which the values of the goods are drawn.
Let us denote the set of all truthful mechanisms by $S_p$ and for each $s_p\in S_p$, let us denote its expected revenue by $o_p(s_p)$.\footnote{As will become clear momentarily, $S$ and $s$ stand for ``solution,'' $p$ stands for ``primal,'' and $o$ stands for ``objective.''} The revenue maximization problem is to find a \emph{solution} $s_p\in S_p$ for which the value of the \emph{objective function} $o_p$ is maximal. \citet{ddt13,ddt15} identify a \emph{dual problem} to the revenue maximization problem: this is a minimization problem, i.e., a problem where the goal is to find a solution $s_d$ from a specific set of feasible solutions $S_d$ that they identify, that minimizes the value $o_d(s_d)$ of a specific objective function $o_d$ that they identify. This problem is an instance of a class of problems called \emph{optimal-transport problems}, and it is a dual problem to the revenue maximization problem in the sense that for every pair of solutions $(s_p,s_d)$ for the primal (revenue maximization) problem and dual (optimal-transport) problem respectively, it holds that the value of the primal objective function for the primal solution is upper-bounded by the value of the dual objective function for the dual solution, that is,
\begin{equation}\label{overview-weak-dual}
\forall s_p\in S_p,s_d\in S_d:\quad o_p(s_p)\le o_d(s_d).
\end{equation}
(See \cref{weak-dual} below for the full details.)
This property is called \emph{weak duality}. \citet{ddt15} also show that this duality is \emph{strong} in the sense that there always exists a pair of solutions $(\hat{s}_p,\hat{s}_d)$ for the primal and dual problems respectively such that 
\begin{equation}\label{overview-tight-dual}
o_p(\hat{s}_p)=o_d(\hat{s}_d).
\end{equation}
A standard observation in duality frameworks is that such $\hat{s}_d$ \emph{certifies} that $\hat{s}_p$ is an optimal solution for the primal problem, since by \cref{overview-weak-dual} the value of \emph{any} primal solution is bounded by $o_d(\hat{s}_d)=o_p(\hat{s}_p)$. Indeed, \citet{ddt15} use their framework to identify and certify optimal primal solutions (revenue-maximizing mechanisms) by identifying such pairs $(\hat{s}_p,\hat{s}_d)$.
To facilitate finding such pairs of solutions, they identify \emph{complementary slackness} conditions, that is, conditions on $s_p$ and $s_d$ that are necessary and sufficient for the inequality in \cref{overview-weak-dual} to in fact be an equality as in \cref{overview-tight-dual}. In particular, for the revenue maximization problem where the two items are sampled i.i.d.\ from the Beta distribution $F=B(1,2)$, they identify such a pair of solutions $(\hat{s}_p,\hat{s}_d)$. They in fact show the that complementary slackness conditions (for this distribution) uniquely define the optimal primal solution and that this solution has an infinite menu-size.

The main idea of our proof is, with the optimal dual $\hat{s}_d$ that \citet{ddt15} identify in hand, to carefully show that for every primal solution $s_p$ \emph{that is a mechanism with small menu-size}, these complementary slackness conditions not only fail to hold for $(s_p,\hat{s}_d)$ (as follows from the result of \citealp{ddt15}), but in fact are sufficiently violated to yield the required separation, that is,
\begin{equation}\label{overview-slackness-violation}
o_p(s_p)< o_d(\hat{s}_d)-\varepsilon = o_p(\hat{s}_p)-\varepsilon,
\end{equation}
where $\hat{s}_p$ is the optimal primal solution and the equality is by \cref{overview-tight-dual}. (See \cref{slackness-violation} below for the full details.)

In slightly more detail, \citet{ddt15} show that for two items sampled i.i.d.\ from $F=B(1,2)$, the complementary slackness conditions dictate that in a certain part of the value space, the set of values of buyers to which Good $2$ (say) is allocated with probability~$1$ by the optimal mechanism and the set of values of buyers to which it is allocated with probability $0$ in that mechanism are separated by a strictly concave curve (the curve $S$ in \cref{gamma} below) and that this implies that the optimal mechanism has an infinite menu-size. Another way to state this conclusion is to observe that it follows from known properties that in a mechanism with a finite menu-size such a curve must be piecewise-linear rather than strictly concave, and so to conclude that the optimal mechanism cannot have a finite menu-size. Roughly speaking, we relate the loss in revenue (compared to the optimal mechanism) of a given mechanism to a certain metric (see below) of the region between the separating curve $S$ of the optimal mechanism and an analogue of this curve (see below for the precise definition) in the given mechanism. We then observe that this analogue of $S$ is a piecewise-linear curve with number of pieces at most the menu-size of the given mechanism, and use this to appropriately lower-bound this metric (see \cref{approx} below) for mechanisms with small menu-size. A lower bound on the loss in revenue for mechanisms with such a menu-size follows.

To present our analysis in further detail, we must now first dive into some of the details of the optimal-transport duality framework of \citet{ddt15}.

\subsection{Minimal Needed Essentials of the Optimal-Transport Duality Framework of \texorpdfstring{\citet{ddt15}}{DDT'15}, and Commentary}\label{sec:lower-ddt}

We now present only the minimal amount of detail from the extensive analysis of \citet{ddt15} that is required to follow our proof; the interested reader is referred to their paper or to the excellent survey of \citet{d15}, whose notation we follow, for the full details that lie beyond the scope of this paper. (See also \citet{gk14} for a slightly different duality approach, and \citet{kf16} for an extension.)

In their analysis, \citet{ddt15} identify a signed Radon measure\footnote{To understand our proof there is no need to be familiar with the general definition of a signed Radon measure. It suffices to know that signed Radon measures generalize distributions that are defined by a combination of atoms and a density function, and allow in particular for \ \ a)~densities (and atoms) that can also be negative (hence the term \emph{signed}), and\ \ b)~the overall measure not necessarily summing up to $1$.}  $\mu$ on $[0,1]^2$ with $\mu\bigl([0,1]^2\bigr)=0$,\footnote{That is, the overall measure sums up to $0$.} such that for a mechanism with utility function\footnote{The utility function of a mechanism maps each buyer type (i.e., pair of valuations) to the utility that a buyer of this type obtains from participating in the mechanism.} $u$, the expected revenue of this mechanism from $F^2$ is equal to\footnote{This \emph{Lebesgue integral} is the measure-theoretic analogue of the expectation (or average) of $u$ with respect to a given distribution (but as $\mu$ is not a distribution, when ``averaging,'' some values are taken with negative weights, and weights do not sum up to $1$).}
\begin{equation}\label{max-u}
\int_{[0,1]^2}ud\mu.
\end{equation}
They show that (the utility function of) the revenue-maximizing mechanism is obtained by maximizing \cref{max-u} subject to $u(0,0)=0$, to $u$ being convex, and to $\forall x,y\in[0,1]^2: u(x)-u(y)\le\bigl|(x-y)_+\bigr|_1$, where $\bigl|(x-y)_+\bigr|_1=\sum_i\max(0,x_i-y_i)$. (\citet{r87} has shown that the utility function of any truthful mechanism satisfies the latter two properties as well as $u(0,0)\ge0$. An equality as in the first property may be assumed without losing revenue or changing the menu-size.)

We comment that while one could have directly attempted prove \cref{lower} by analyzing how much revenue is lost in \cref{max-u} due to restricting attention only to $u$ that corresponds to a mechanism with a certain (small) menu-size (in particular, the graph of such $u$ is a maximum of $C$ planes, where $C$ is the menu-size), such an analysis, even if successful, would have been hard and involved, and immensely tailored to the specifics of the distribution~$F^2$, due to the complex definition of $\mu$. For this reason we base our analysis on the duality-based framework of \citet{ddt13,ddt15}, which they have developed to help find and certify the optimal $u$, and we show, for the first time to the best of our knowledge, how to use this framework to quantitatively reason about the revenue loss from suboptimal mechanisms. The resulting approach is principled, general, and robust.\footnote{For example, readers familiar with the definition of \emph{exclusion set} mechanisms \citep{ddt15} may notice that our analysis of $F^2$ below can be readily applied with virtually no changes also to other distributions for which the optimal mechanism is derived from an exclusion set that is nonpolygonal (as $\mathcal{Z}$ is in the analysis below).}

\enlargethispage{.657em}
\citet{ddt15} show that for every utility function~$u$ of a (truthful) mechanism for valuations in $[0,1]^2$ and for every coupling\footnote{Informally (and sufficient to understand our proof), a coupling $\gamma$ of two unsigned Radon measures $\mu_1$ and $\mu_2$ both having the same overall measure is a recipe for rearranging the mass of $\mu_1$ into the mass of $\mu_2$ by specifying where each piece of (positive) mass is transported. Formally, a coupling $\gamma$ of two unsigned Radon measures $\mu_1$ and $\mu_2$ on $[0,1]^2$ with $\mu_1\bigl([0,1]^2\bigr)=\mu_2\bigl([0,1]^2\bigr)$ is an unsigned Radon measure on $[0,1]^2\times[0,1]^2$ whose marginals are $\mu_1$ and $\mu_2$, i.e., for every measurable set $A\subseteq[0,1]^2$, it holds that
$\gamma\bigl(A\times[0,1]^2\bigr)=\mu_1(A)$ and $\gamma\bigl([0,1]^2\times A\bigr)=\mu_2(A)$.}
$\gamma$ of $\mu'_+$ and $\mu'_-$, where $\mu'=\mu'_+-\mu'_-$ is any measure that convex-dominates\footnote{A distribution $\mu'$ convex-dominates a distribution $\mu$ if $\mu'$ is obtained from $\mu$ by shifting mass to coordinate-wise larger points and by performing mean-preserving spreads of positive mass. To follow our paper only a single property of convex dominance is needed --- see below. As \citet{ddt15} show (but not required for our proof), of interest in this context are in fact only cases where $\mu'$ is obtained from~$\mu$ by mean-preserving spreads of positive mass.} $\mu$, it is the case that\footnote{Once again, the Lebesgue integral on the right-hand side is the measure-theoretic analogue of the average of the values $\bigl|(x-y)_+\bigr|_1$, where, informally, each pair $(x,y)$ is taken with weight equal to the amount of (positive) mass transported by $\gamma$ from $x$ to~$y$.}
\begin{equation}\label{weak-dual}
\int_{[0,1]^2}ud\mu\le
\int_{[0,1]^2\times[0,1]^2}\bigl|(x-y)_+\bigr|_1d\gamma(x,y).\pagebreak 
\end{equation}
(This is precisely \cref{overview-weak-dual} in full detail.)
They identify the optimal mechanism $\hat{M}$ for the distribution $F^2$ by finding a measure\footnote{For this specific distribution $F^2$, the measure $\hat{\mu}'$  that they identify in fact equals $\mu$.}~$\hat{\mu}'$ and a coupling $\hat{\gamma}$ of $\hat{\mu}'_+$ and $\hat{\mu}'_-$, such that \cref{weak-dual} holds with an equality for $u=\hat{u}$ (the utility function of $\hat{M}$) and $\gamma=\hat{\gamma}$ (this is precisely \cref{overview-tight-dual} in full detail).\footnote{In fact, \citet{ddt15} prove a beautiful theorem that states that this (i.e., finding suitable $\mu'$ and $\gamma$ such that \cref{weak-dual} holds with an equality for the optimal~$u$) can be done for any underlying distribution, i.e., that this duality is \emph{strong}.} To find $\hat{u}$ and $\hat{\gamma}$, they make use of \emph{complementary slackness conditions} that they identify, and make sure they are all completely satisfied. In our proof below, we will claim that for any utility function $u$ that corresponds to a mechanism with small menu-size, the complementary slackness conditions with respect to $u$ and $\hat{\gamma}$ will be sufficiently violated so as to give sufficient separation between the left-hand side of \cref{weak-dual} for $u$ (that is, the revenue of the mechanism with small menu-size) and the right-hand side of \cref{weak-dual} for~$\hat{\gamma}$ (that is, the optimal revenue).

\looseness=-1
To better understand the complementary slackness conditions identified by \citet{ddt15}, let us review their proof for \cref{weak-dual}:
\begin{multline}\label{slackness}
\int_{[0,1]^2}ud\mu\le
\int_{[0,1]^2}ud\mu'=
\int_{[0,1]^2}ud(\mu'_+-\mu'_-)=
\int_{[0,1]^2\times[0,1]^2}\bigl(u(x)-u(y)\bigr)d\gamma(x,y)\le\\
\le\int_{[0,1]^2\times[0,1]^2}\bigl|(x-y)_+\bigr|_1d\gamma(x,y),
\end{multline}
where the first inequality is since $u$ is convex (this inequality is the only property of convex dominance that is needed to follow our paper), the second equality is by the definition of a coupling, and the second inequality is due to the third property of $u$ as defined above following \cref{max-u}. \citet{ddt15} note that if it would not be the case that $\gamma$-almost everywhere\footnote{Informally (and sufficient to understand our proof), for a condition to hold \emph{$\gamma$-almost everywhere} means for that condition to hold for every $x$ and $y$ such that the coupling~$\gamma$ transports (positive) mass from $x$ to $y$.} we would have $u(x)-u(y)=\bigl|(x-y)_+\bigr|_1$, then the second inequality in \cref{slackness} would be strict, and so the same proof would give a strict inequality in \cref{weak-dual}; they use this insight to guide their search for the optimal $\hat{u}$ and its tight dual~$\hat{\gamma}$. (They also perform a similar analysis with respect to the first inequality in \cref{slackness}, which we skip as we do not require it.) In our proof below we will show that for the coupling $\hat{\gamma}$ that they identify, and for any $u$ that corresponds to a mechanism with small menu-size, this constraint (i.e., that $\gamma$-almost everywhere $u(x)-u(y)=\bigl|(x-y)_+\bigr|_1$) will be significantly violated, in a precise sense. To do so, we first describe the measure $\hat{\mu}'$ and the coupling $\hat{\gamma}$ that they identify.

Examine \cref{gamma}.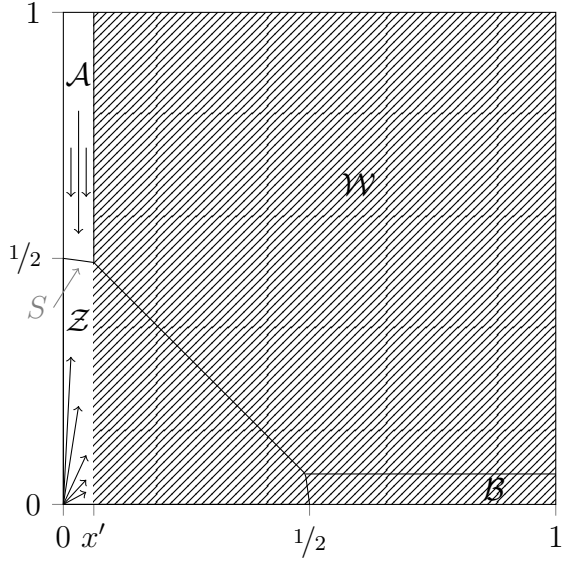
\begin{figure}%
\centering%
\begin{tikzpicture}
\begin{axis}[height=3.2in, width=3.2in, ymin=0, ymax=1, xmin=0, xmax=1, xtick pos=left, xtick align=outside, xtick={0,0.5,1}, xticklabels={$0$,$\nicefrac{1}{2}$, $1$}, ytick pos=bottom, ytick align=outside, ytick={0,0.5,1}, yticklabels={0,$\nicefrac{1}{2}$,$1$}, extra x ticks={0.06187679}, extra x tick label={\smash{$x'$}\vphantom{$0$}}, clip=false]
\addplot[color=black, mark=none, samples=200, domain=0:0.5] {min( (2-3*x)/(4-5*x), (2-4*x)/(3-5*x), 0.5534938-x)}\closedcycle; 
\def\xprime{0.06187679}
\draw (axis cs:.5534938-\xprime,\xprime,) -- (axis cs:1,0.06187679); 
\draw (axis cs:\xprime,.5534938-\xprime) -- (axis cs:0.06187679,1); 

\fill[pattern=north east lines] (axis cs:\xprime,1) -- (axis cs:1,1) -- (axis cs:1,0) -- (axis cs:\xprime,0) -- cycle;

\draw[<-,gray,overlay] (axis cs:\xprime/2,0.48) -- (-.02,.4) node[left,shift={(0.01,0)}] {$S$};
\node at (axis cs:0.875,0.03){$\mathcal{B}$};
\node at (axis cs:0.6,0.65){$\mathcal{W}$};
\node at (axis cs:0.03,0.375){$\mathcal{Z}$};

\node at (axis cs:0.03,0.875){$\mathcal{A}$};
\draw [->] (axis cs:\xprime/2,0.8) -- ++(axis cs:0,-0.25);
\draw [->] (axis cs:\xprime/4,0.725) -- ++(axis cs:0,-0.1);
\draw [->] (axis cs:\xprime*3/4,0.725) -- ++(axis cs:0,-0.1);

\draw[->] (axis cs:0,0) -- (axis cs:\xprime/4,.3);
\draw[->] (axis cs:0,0) -- (axis cs:\xprime/2,.2);
\draw[->] (axis cs:0,0) -- (axis cs:\xprime*3/4,.1);
\draw[->] (axis cs:0,0) -- (axis cs:\xprime*3/4,.05);
\draw[->] (axis cs:0,0) -- (axis cs:\xprime*3/4,.025);
\end{axis}
\end{tikzpicture}%
\caption{\label{gamma}The optimal coupling $\hat{\gamma}$ identified by \citet{ddt15} for $F^2$, illustrated in the region $R=[0,x']\times[0,1]$. \namecref{gamma} adapted from \citet{ddt15} with permission.}\end{figure}
For our proof below it suffices to describe the measure $\hat{\mu}'$ and the coupling $\hat{\gamma}$, both restricted to a region $R\eqdef[0,x']\times[0,1]$ for an appropriate $x'>0$.\footnote{We choose $x'$ as the horizontal-axis coordinate of the right boundary of the region denoted by $\mathcal{A}$ in \citet{ddt15}.} The measure $\hat{\mu}'$ has a point mass of measure~$1$ in $(0,0)$, and otherwise in every $(x_1,x_2)\in R\setminus\bigl\{(0,0)\bigr\}$ has density
\[
f(x_1)f(x_2)\left(\frac{1}{1-x_1}+\frac{1}{1-x_2}-5\right).
\]
In the region $\mathcal{A}$, the coupling $\hat{\gamma}$ sends positive mass downward, from positive-density points to negative-density points. In the region $\mathcal{Z}$, the coupling $\hat{\gamma}$ sends positive mass from $(0,0)$ upward and rightward to all points (the density is indeed negative throughout $\mathcal{Z}\setminus\bigl\{(0,0)\bigr\}$). No other positive mass is transported inside $R$ or into~$R$. (Some additional positive mass from $(0,0)$ is transported out of $R$.)

The optimal mechanism that \citet{ddt15} identify does not award any good (nor does it charge any price) in the region $\mathcal{Z}$, while awarding Good~$2$ with probability~$1$ and Good~$1$ with varying probabilities (and charging varying prices) in the region $\mathcal{A}$.

As \citet{ddt15} note, indeed $\hat{\gamma}$-almost everywhere the complementary slackness condition ${\hat{u}(x)-\hat{u}(y)}=\bigl|(x-y)_+\bigr|_1$ holds for this mechanism: in the region $\mathcal{A}$, coupled points~$x,y$ have ${x_1=y_1}$ and $x_2>y_2$, and in this region, $\hat{u}'_2=1$; in the region~$\mathcal{Z}$, coupled points have $x_i=0\le y_i$ and $\hat{u}(x)=\hat{u}(y)$. In fact, this reasoning shows that given the optimal coupling $\hat{\gamma}$, the utility function~$\hat{u}$ is uniquely defined by the complementary slackness conditions, and so is the unique revenue-maximizing utility function. Since it is well known that wherever a utility function $u(x_1,x_2)$ of a truthful mechanism is differentiable, its derivative in the direction of $x_i$ is the allocation probability of Good $i$ (indeed, by examining Good~$i=2$ one can verify using this property that the mechanism corresponding to $\hat{u}$ indeed awards Good $2$ with probability~$1$ in the region~$\mathcal{A}$ and with probability $0$ in the region~$\mathcal{Z}$), then by examining Good~$i=1$, since the curve $S$ (see \cref{gamma}) that separates the regions $\mathcal{Z}$ and~$\mathcal{A}$ is strictly concave, \citet{ddt15} conclude that there is a continuum of allocation probabilities of Good $1$ in the mechanism corresponding to $\hat{u}$ (which is the unique revenue-maximizing mechanism), thus concluding that the unique revenue-maximizing mechanism for the distribution $F^2$ has an infinite menu-size.

\subsection{Proof of Theorem~\ref{lower}}\label{sec:lower-proof}

An alternative way to state the conclusion of the argument of \citet{ddt15} that any revenue-maximizing mechanism for $F^2$ has an infinite menus-size is as follows: let $u$ be the utility function of a mechanism $M$ with a finite menu-size. It is well known that the graph of $u$ is the maximum of a finite number of planes (each corresponding to one entry in the menu of $M$). Therefore, since $S$ is strictly concave, it is impossible for $u$ to have derivative~$0$ beneath the curve $S$ and derivative $1$ above the curve $S$, and so the complementary slackness conditions must be violated, and hence $M$ is not optimal. In our proof we will quantify the degree of violation of the complementary slackness conditions as a function of the finite menu-size of such $u$. We would like to reason as follows: for such $u$ with a finite menu-size, define the corresponding curve~$T$ that is the analogue for $u$ of the curve $S$, and then show that since $T$ must be piecewise-linear, quantifiable revenue is lost due to the complementary slackness conditions not holding in the region between $S$ and $T$. It is not immediately clear how to define $T$, though.

Intuitively we would have liked to define $T$ to be the curve on $[0,x']$ above which $u$ awards Good $2$ with probability $1$ and below which $u$ awards Good $2$ with probability $0$, but what if $u$ also awards Good $2$ with fractional probability? How should we define $T$ in such cases? (Remember that all that we know about $u$ is that it has small menu-size.) As we will see below, to show that we indeed have quantifiably sufficient revenue loss from any deviation of $T$ from~$S$, we will define $T$ as the curve above which $u$ awards Good $2$ with probability more than one half, and below which $u$ awards Good~$2$ with probability less than or equal to one half. As will become clear from our calculations, the constant one half could have been replaced here with any fixed fraction,\footnote{Similarly, the direction of tie breaking with respect to the region where Good $2$ is awarded with probability precisely one half could have been flipped.} but crucially it could not have been replaced with $0$ (i.e., defining $T$ as the curve above which $u$ awards Good $2$ with positive probability and below which $u$ does not award Good $2$) or with $1$ (i.e., defining $T$ as the curve above which $u$ awards Good $2$ with probability $1$ and below which $u$ awards Good $2$ with probability strictly less than $1$).

We are now finally ready to prove \cref{lower}, with the help of the following \lcnamecref{approx}, which may be of separate interest. (We prove \cref{approx} following the proof of \cref{lower}.)

\begin{proposition}\label{approx}
Let $S:[0,x']\rightarrow\mathbb{R}$ be a concave function with radius of curvature at most~$r$ everywhere, for some $r<\infty$. For small enough $\delta$, the following holds. For any piecewise-linear function~${T:[0,x']\rightarrow\mathbb{R}}$ composed of at most $\frac{x'}{8\sqrt{r\delta}}$ linear segments, the Lebesgue measure of the set of coordinates $x_1$ with $|S(x_1)-T(x_1)|>\delta$ is at least $\nicefrac{x'}{2}$.
\end{proposition}

\begin{proof}[Proof of \cref{lower}]
The curve $S$ (see \cref{gamma}) that separates the regions $\mathcal{Z}$ and $\mathcal{A}$ is given by $x_2=\frac{2-3x_1}{4-5x_1}$ (where $x_1\in[0,x']$) \citep{ddt15}. Therefore, it is strictly concave, having radius of curvature at most $r$ everywhere, for some fixed $r<\infty$. We note also that there exists a constant $d>0$ and a neighborhood~$N$ of the curve~$S$ in~$R$ in which the density of $\hat{\mu}'$ is (negative and) smaller than $-d$.

Let $\varepsilon>0$ and set $\delta\eqdef\sqrt{\frac{8\varepsilon}{x'\cdot d}}$. Assume without loss of generality that $\varepsilon$ is sufficiently small so that both\ \ i)~the $\delta$-neighborhood of $S$ in $R$ is contained in the neighborhood~$N$ of~$S$, and\ \ ii)~\cref{approx} holds with respect to $\delta$. Let $C\eqdef\frac{x'}{8\sqrt{r\delta}}=\Omega(\nicefrac{1}{\sqrt[4]{\varepsilon}})$.

Let $u$ be the utility function of a mechanism $M$ with menu-size at most $C$, and let $T:[0,x']\rightarrow[0,1]$ be defined as follows:
\[
T(x_1)\eqdef
\inf\bigl\{x_2\in[0,1]~\big|~u'_2(x_1,x_2)>\nicefrac{1}{2}\bigr\}=
\sup\bigl\{x_2\in[0,1]~\big|~u'_2(x_1,x_2)\le\nicefrac{1}{2}\bigr\}.
\]
 It is well known that the graph of $u$ is the maximum of $C$ planes. Therefore, $T$ is a piecewise-linear function composed of at most $C$ linear segments.

Let $y_1\in[0,x']$ with $T(y_1)-S(y_1)>\delta$. Let $y_2\in\bigl(S(y_1),S(y_1)\!+\!\nicefrac{\delta}{2}\bigr)$ and let $x_2$ be such that $\hat{\gamma}$ transfers positive mass from $x\eqdef(y_1,x_2)$ to $y\eqdef(y_1,y_2)$. (All mass transferred to $y$ by~$\hat{\gamma}$ is from points of this form.) We note that by definition of $T$,
\[
u(x)-u(y)\le
x_2-y_2-\nicefrac{\delta}{4}=
\bigl|(x-y)_+\bigr|_1-\nicefrac{\delta}{4}.
\]

Similarly, let $y_1\in[0,x']$ with $S(y_1)-T(y_1)>\delta$. Let $y_2\in{\bigl(S(y_1)\!-\!\nicefrac{\delta}{2},S(y_1)\bigr)}$. Note that~$\hat{\gamma}$ transfers positive mass from $x\eqdef(0,0)$ to $y\eqdef(y_1,y_2)$. (All mass transferred to $y$ by~$\hat{\gamma}$ is from the point $x$.) We note that by definition of $T$,
\[
u(x)-u(y)\le
-\nicefrac{\delta}{4}=
\bigl|(x-y)_+\bigr|_1-\nicefrac{\delta}{4}.
\]

By \cref{approx}, the Lebesgue measure of the set of coordinates $y_1$ with $\bigl|S(y_1)-T(y_1)\bigr|>\delta$ is at least $\nicefrac{x'}{2}$.
Similarly to \cref{slackness}, we therefore obtain
\begin{multline}\label{slackness-violation}
\Rev_{F^2}(M)=\\=
\int_{[0,1]^2}ud\mu\le
\int_{[0,1]^2}ud\hat{\mu}'=
\int_{[0,1]^2}ud(\hat{\mu}'_+-\hat{\mu}'_-)=
\int_{[0,1]^2\times[0,1]^2}\bigl(u(x)-u(y)\bigr)d\hat{\gamma}(x,y)\le\\\le
\int_{[0,1]^2\times[0,1]^2}\bigl|(x-y)_+\bigr|_1d\hat{\gamma}(x,y)-\nicefrac{\delta}{4}\cdot\nicefrac{\delta}{2}\cdot\nicefrac{x'}{2}\cdot d=
\int_{[0,1]^2}\hat{u}d\mu-\nicefrac{\delta}{4}\cdot\nicefrac{\delta}{2}\cdot\nicefrac{x'}{2}\cdot d=\\=
\OPT(F^2)-\nicefrac{\delta}{4}\cdot\nicefrac{\delta}{2}\cdot\nicefrac{x'}{2}\cdot d=
\OPT(F^2)-\nicefrac{\varepsilon}{2}<
\OPT(F^2)-\varepsilon,
\end{multline}
as required (this is precisely \cref{overview-slackness-violation} in full detail).
\end{proof}

\begin{proof}[Proof of \cref{approx}]
We will show that from each linear segment of $T$, at most a Lebesgue measure $4\sqrt{r\delta}$ of coordinates $x_1$ satisfy
\begin{equation}\label{chord}
\bigl|S(x_1)-T(x_1)\bigr|\le\delta.
\end{equation}
This implies the \lcnamecref{approx} since this means that from all linear segments of $T$ together, at most a Lebesgue measure $\frac{x'}{8\sqrt{r\delta}}\cdot4\sqrt{r\delta}=\frac{x'}{2}$ of coordinates  $x_1$ satisfy \cref{chord}, and hence at least a Lebesgue measure $\frac{x'}{2}$ of coordinates $x_1$ satisfy $\bigl|S(x_1)-T(x_1)\bigr|>\delta$, as required.

For a Lebesgue measure $m$ of coordinates from a single linear segment of $T$ to satisfy \cref{chord}, we note that a necessary condition is that $m$ be at most the length of a chord of sagitta at most~$2\delta$ in a circle of radius at most~$r$. (See \cref{sagitta}.)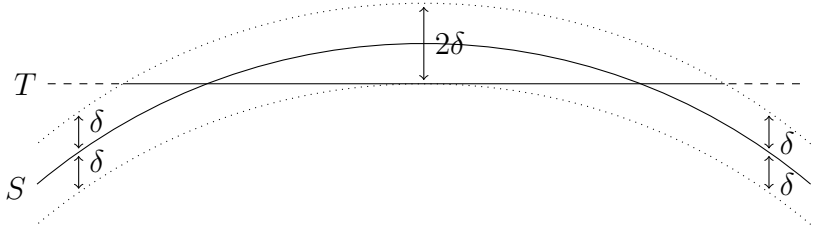
\begin{figure}%
\centering%
\begin{tikzpicture}
\def\sagitta{1.07179676972}
\draw[dotted] (50:8) arc (50:130:8);
\draw (50:8) ++(0,-\sagitta/2) arc (50:130:8) node[left,shift={(0,-0.05)}] {$S$};
\draw[dotted] (50:8) ++(0,-\sagitta) arc (50:130:8);
\draw (-4,8-\sagitta) -- (4,8-\sagitta);
\draw[dashed] (-5,8-\sagitta) -- (-4,8-\sagitta) ++(-1,0) node[left] {$T$};
\draw[dashed] (5,8-\sagitta) -- (4,8-\sagitta);
\draw[<->] (0,8.05-\sagitta) -- (0,7.95);
\draw (0,8-\sagitta/2) node[right] {$2\delta$};
\draw[<->] (125:8) ++(0,-0.05) -- ++(0,-\sagitta/2+0.1);
\draw (125:8) ++ (0,-\sagitta/4+.15) node[right] {$\delta$};
\draw[<->] (125:8) ++(0,-\sagitta/2-0.05) -- ++(0,-\sagitta/2+0.1);
\draw (125:8) ++ (0,-\sagitta*3/4+.15) node[right] {$\delta$};
\draw[<->] (55:8) ++(0,-0.05) -- ++(0,-\sagitta/2+0.1);
\draw (55:8) ++ (0,-\sagitta/4-.15) node[right] {$\delta$};
\draw[<->] (55:8) ++(0,-\sagitta/2-0.05) -- ++(0,-\sagitta/2+0.1);
\draw (55:8) ++ (0,-\sagitta*3/4-.15) node[right] {$\delta$};
\end{tikzpicture}%
\caption{\label{sagitta}In the extreme case where $S$ is an arc of radius $r$, a line segment with the maximum Lebesgue measure of coordinates satisfying \cref{chord} is a horizontal chord of sagitta~$2\delta$ in a circle (the circle of the top dotted arc) of radius~$r$.}\end{figure}
We claim that this implies that $m\le\sqrt{16r\delta-16\delta^2}\le4\sqrt{r\delta}$, as required. Indeed, in the extreme case where $m$ is the length of a chord of sagitta precisely~$2\delta$ in a circle of radius precisely~$r$, we have by a standard use of the Intersecting Chords Theorem\footnote{The Intersecting Chords Theorem states that when two chords of the same circle intersect, the product of the lengths of the two segments (that are delineated by the intersection point) of one chord equals the product of the lengths of the two segments of the other.} that $(\frac{m}{2})^2=(2r-2\delta)\cdot2\delta$. Solving for~$m$, we have that (in the extreme case) $m=\sqrt{16r\delta-16\delta^2}$, as claimed.
\end{proof}

\section{Upper Bound}\label{sec:upper}

Recall that \cref{hn}, due to \citet{hn13}, provides an upper bound of $O(\nicefrac{1}{\varepsilon^4})$ on the menu-size of some mechanism that maximizes revenue up to an additive $\varepsilon$. Their proof uses virtually no information regarding the structure of the optimal mechanism: it starts with a revenue-maximizing mechanism, and cleverly rounds two of the three coordinates (allocation probability of Good~$1$, allocation probability of Good~$2$, price) of every outcome, to obtain a mechanism with small menu-size without significant revenue loss. We will follow a similar strategy, but will only round one of these three coordinates (namely, the price), using a result by \citet{p11} that shows that under an assumption on distributions that is standard in the economics literature on multidimensional mechanism design \citep[see, e.g.,][]{mm88,mv06,p11}, for an appropriate choice of revenue-maximizing mechanism, one of the other (allocation) coordinates is in fact already rounded (specifically, it is either zero or one).

\begin{definition}[\citealp{mm88}]\label{hazard}
A distribution $F\in\Delta\bigl([0,1]^n\bigr)$ is said to satisfy the \emph{McAfee-McMillan hazard condition} if it has a differentiable density function $f$ satisfying
\[
(n+1)f(x)+x\cdot\nabla f(x)\ge0
\]
for every $x=(x_1,\ldots,x_n)\in[0,1]^n$.
\end{definition}

\begin{theorem}[\citealp{p11}]\label{exposed}
For every distribution $F\in\Delta\bigl([0,1]^2\bigr)$ satisfying the McAfee-McMillan hazard condition (for $n=2$), there exists a mechanism $M$ that maximizes the revenue from $F$ and has no outcome for which both allocation probabilities are in the open interval $(0,1)$.
\end{theorem}

Our upper bound follows by combining the argument of \citet{hn13} with \cref{exposed}.

\begin{proof}[Proof of \cref{upper}]
Let $M$ be a revenue-maximizing mechanism for $F$ as in \cref{exposed}.
Let $\delta\eqdef\varepsilon^2$, and for every real number $t$, denote by $\deltafloor{t}$  the rounding-down of $t$ to the nearest integer multiple of~$\delta$.
Let $M'$ be the mechanism whose menu is comprised of all outcomes of the form $(p_1,p_2;(1\!-\!\varepsilon)\!\cdot\!\deltafloor{t})$ for every outcome $(p_1,p_2;t)$ of $M$ (where $p_i$ is the allocation probability of Good~$i$, and $t$ is the price charged in this outcome). We claim that $\Rev_F(M')>(1-\varepsilon)\cdot\Rev_F(M)-\varepsilon\ge\Rev_F(M)-2\varepsilon$.

The main idea of the ``nudge and round'' argument of \citet{hn13} is that while the rounding (which is performed to reduce the menu-size --- see below), by itself (without the discounting, which is the ``nudge'' part), could have hypothetically constituted the ``last straw'' that causes some buyer type to switch from preferring a very expensive outcome to preferring a very cheap one (thus significantly hurting the revenue), since more expensive outcomes are more heavily discounted, then this compensates for any such ``last straw'' effects. More concretely, while the rounding, before the $\varepsilon$-discounting, can cause a buyer's utility from any outcome to shift by at most $\varepsilon^2$ (which could be the ``last straw''), and since for any outcome whose price is cheaper by more than an $\varepsilon$ compared to the buyer's original outcome of choice the given discount is smaller by at least $\varepsilon^2$, this smaller discount more than eliminates any potential utility gain due to rounding, so such an outcome would not become the most-preferred one. We will now formally show this.

Fix a type $x=(x_1,x_2)\in[0,1]^2$ for the buyer. Let $e$ be the outcome according to~$M$ when the buyer has type $x$. It is enough to show that the buyer pays at least $(1\!-\!\varepsilon)t_e-\varepsilon$ according to $M'$ when she has type $x$. (We denote the price of, e.g., $e$ by $t_e$.) Let $f'$ be a possible outcome of $M'$, and let~$f$ be the outcome of $M$ that corresponds to it. We will show that if $(1\!-\!\varepsilon)\deltafloor{t_f}<(1\!-\!\varepsilon)(\deltafloor{t_e}-\varepsilon)$, then a buyer of type $x$ strictly prefers the outcome $e'$ of $M'$ that corresponds to $e$ over~$f'$ (and so does not choose $f'$ in $M'$). Indeed, since in this case $\deltafloor{t_f}<\deltafloor{t_e}-\varepsilon$, we have that
\begin{multline*}
u_{x}(e')\ge
u_{x}(e)+\varepsilon\cdot \deltafloor{t_e}\ge
u_{x}(f)+\varepsilon\cdot \deltafloor{t_e}>
u_{x}(f')-\delta-\varepsilon\cdot\deltafloor{t_f}+\varepsilon\cdot \deltafloor{t_e}=\\=
u_{x}(f')-\delta+\varepsilon\cdot\bigl(\deltafloor{t_e}-\deltafloor{t_f}\bigr)>
u_{x}(f')-\delta+\varepsilon\cdot \varepsilon=
u_{x}(f'),
\end{multline*}
so in $M'$, a buyer of type $x$ pays at least
\[
(1\!-\!\varepsilon)(\deltafloor{t_e}-\varepsilon)>
(1\!-\!\varepsilon)(t_e-\delta-\varepsilon)=
(1-\varepsilon)(t_e-\varepsilon^2-\varepsilon)>
(1-\varepsilon)t_e-\varepsilon.
\]
How many menu entries does $M'$ really have (i.e., how many menu entries ever get chosen by the buyer)? The number of menu entries $(p_1,p_2;t)$ with $p_1=1$ is at most $O(\nicefrac{1}{\varepsilon^2}$), since for every price $t$ (there are $O(\nicefrac{1}{\varepsilon^2})$ such options) we can assume without loss of generality that only the menu entry $(p_1,p_2;t)$ with highest $p_2$ will ever be chosen.\footnote{If a maximum such $p_2$ is not attained, then we can add a suitable menu entry with the supremum of such~$p_2$; see \citet{bgn17} for a full argument.} A similar argument for the cases $p_1=0$, $p_2=0$, and $p_2=1$ (by \cref{exposed}, no more cases exist beyond these four) shows that in total there really are at most  $O(\nicefrac{1}{\varepsilon^2})$ menu entries in~$M'$, as required.
\end{proof}

We note that both \citet{p11} and \citet{kf16} conjecture that the conclusion of \cref{exposed} holds under more general conditions than the McAfee-McMillan hazard condition. An affirmation of (either of) these conjectures would, by the above proof, immediately imply that the conclusion of \cref{upper} holds under the same generalized assumptions.

\citet{hn13} also analyze the scenario of a two-good distribution supported on $[1,H]^2$ for any given $H$, and give an upper bound of $O(\nicefrac{\log^2 H}{\varepsilon^5})$ on the menu-size that suffices for revenue maximization up to a \emph{multiplicative} $\varepsilon$. Using the above techniques, their argument could be similarly modified to give an improved upper bound of $O(\nicefrac{\log H}{\varepsilon^2})$ in that setting for distributions $F\in\Delta\bigl([1,H]^2\bigr)$ satisfying the McAfee-McMillan hazard condition (or any generalized condition under which the conclusion of \cref{exposed} holds).

\section{Extensions}\label{extensions}

\subsection{More than Two Goods}

Recall that \cref{cc} concludes, from our menu-size lower bound (\cref{lower}) and the menu-size upper bound of \citet{hn13} (\cref{hn}), a tight bound on the minimum deterministic communication complexity guaranteed to suffice for running some up-to-$\varepsilon$ revenue-maximizing mechanism for selling two goods, thereby completely resolving this problem. In fact, since \citet{hn13} prove an upper bound of $(\nicefrac{n}{\varepsilon})^{O(n)}$ (later strengthened to $(\nicefrac{\log n}{\varepsilon})^{O(n)}$ by \citealp{dhn14}) on the menu-size required for revenue maximization up to an additive $\varepsilon$ when selling any number of goods $n$, we obtain our tight communication-complexity bound not only for two goods, but for any fixed number of goods~$n\ge2$:

\begin{corollary}[Communication Complexity: Tight Bound for Any Number of Goods]\label{cc-add}
Fix $n\ge2$. There exists $D_n(\varepsilon)=\Theta(\log\nicefrac{1}{\varepsilon})$ such that for every $\varepsilon>0$ it is the case that $D_n(\varepsilon$) is the minimum communication complexity that satisfies the following: for every distribution $F\in\Delta\bigl([0,1]^n\bigr)$ there exists a mechanism~$M$ for selling $n$ goods such that the deterministic communication complexity of running~$M$ is $D_n(\varepsilon$) and such that $\Rev_F(M)>\OPT(F)-\varepsilon$. This continues to hold even if $F$ is guaranteed to be a product distribution.
\end{corollary}

For the case of one good, the seminal result of \citet{m81} shows that there exists a (precisely) revenue-maximizing mechanism with only two possible outcomes (and hence deterministic communication complexity of $1$), which simply offers the good for a suitably chosen take-it-or-leave-it price. \cref{cc-add} therefore shows a precise dichotomy in the asymptotic communication complexity of up-to-$\varepsilon$ revenue maximization, between the case of one good (\citeauthor{m81}'s result; $1$ bit of communication) on the one hand, and the case of any other fixed number of goods ($\Theta(\log\nicefrac{1}{\varepsilon})$ bits of communication) on the other hand.

\subsection{Multiplicative Approximation}

In a scenario where the valuations may be unbounded, i.e., $v_i\in[0,\infty)$ for all $i$, \citet{hn13} have shown that no finite menu-size suffices for maximizing revenue up to a multiplicative\footnote{For unbounded valuations, it makes no sense to consider additive guarantees, as the problem is invariant under scaling of the currency.}~$\varepsilon$, and consequently \citet{hn14} asked\footnote{\citet{hn14} is a manuscript combining \citet{hn13} with an earlier paper.} whether this impossibility may be overcome for the case of independently distributed valuations for the goods. \cite*{bgn17} gave a positive answer, showing that for every $n$ and~$\varepsilon$, a finite menu-size suffices, and moreover gave an upper bound of $(\nicefrac{\log n}{\varepsilon})^{O(n)}$ on the sufficient menu-size. Since for valuations in $[0,1]$, revenue maximization up to a multiplicative~$\varepsilon$ is a stricter requirement than revenue maximization up to an additive~$\varepsilon$, our lower bound from \cref{lower} immediately provides a lower bound for this scenario as well.

\begin{corollary}[Menu-Size for Multiplicative Approximation: Lower Bound]\label{lower-mult}
There exist $C(\varepsilon)=\Omega(\nicefrac{1}{\sqrt[4]{\varepsilon}})$ and a distribution $F\in\Delta\bigl([0,1]\bigr)\subseteq\Delta\bigl([0,\infty)\bigr)$, such that for every $\varepsilon>0$ it is the case that $\Rev_{F^2}(M)<(1-\varepsilon)\cdot\OPT(F^2)$ for every mechanism $M$ with menu-size at most $C(\varepsilon)$.
\end{corollary}

By an argument similar to that yielding \cref{cc-add}, using the above upper bound of \citet{bgn17} in lieu of that of \citet{hn13} / \citet{dhn14}, we therefore obtain an analogue of \cref{cc-add} for this setting as well.

\begin{corollary}[Communication Complexity of Multiplicative Approximation: Tight Bound]\label{cc-mult}
Fix $n\ge2$. There exists $D_n(\varepsilon)=\Theta(\log\nicefrac{1}{\varepsilon})$ such that for every $\varepsilon>0$ it is the case that $D_n(\varepsilon)$ is the minimum communication complexity that satisfies the following: for every product distribution $F\in\Delta\bigl([0,\infty)\bigr)^n$ there exists a mechanism~$M$ for selling $n$ goods such that the deterministic communication complexity of running~$M$ is $D_n(\varepsilon$) and such that ${\Rev_F(M)>(1-\varepsilon)\cdot\OPT(F)}$. This continues to hold even if each of the distributions whose product is $F$ is guaranteed to be supported in $[0,1]$.
\end{corollary}

\cref{cc-mult} shows that the dichotomy between one good and any other fixed number of goods that is shown by \cref{cc-add} also holds for multiplicative up-to-$\varepsilon$ revenue maximization.

\section{Discussion}\label{discussion}

In a very recent paper, \citet{ssw17} analyze the menu-size required for approximate revenue maximization in what is known as the FedEx problem \citep{fgkk16}. Interestingly, they also use piecewise-linear approximation of concave functions to derive their bounds. Nonetheless, there are considerably many differences between their analysis and ours, e.g.: the effects of bad piecewise-linear approximation on the revenue,\footnote{There: even a single point having large distance can cause quantifiable revenue loss; Here: at least a certain measure of points having large distance causes quantifiable revenue loss.}
the approximated/approximating object,\footnote{There: revenue curves, with vertices corresponding to menu entries; Here: contour lines of the allocation function, with edges corresponding to menu entries.} the mathematical features of the approximated object,\footnote{There: piecewise-linear; Here: strictly concave.} the geometric/analytic proof of the bound on the number of linear segments, and finally, the argument that uses the piecewise-linear approximation and whether or not it couples the desired approximation with another parameter of the problem.\footnote{There: lower bound achieved by coupling the desired approximation with the number of possible deadlines~$n$ (setting $\varepsilon\eqdef\nicefrac{1}{n^2}$); Here: desired approximation uncoupled from the number of goods $n$.} It therefore seems to be unlikely that both analyses are special cases of some general analysis, and so it would be interesting to see whether piecewise-linear approximations of concave (or other) functions ``pop up'' in the future in any additional contexts in connection with bounds on the menu-size of mechanisms.\footnote{Incidentally, other known derivations of menu-size bounds, such as those in \citet{hn13}, \citet{dhn14}, and \citet{bgn17}, do not use (even implicitly, to the best of our understanding) piecewise-linear approximations.}

\looseness=-1 As mentioned in \cref{extensions}, \citet{bgn17} prove an upper bound on the menu-size required for multiplicative up-to-$\varepsilon$ approximate revenue maximization when selling $n$ goods to an additive buyer with independently distributed valuations. For any fixed $n$, this bound is polynomial in $\nicefrac{1}{\varepsilon}$, and the lower bound that we establish in \cref{lower-mult} shows in particular that a polynomial dependence cannot be avoided here (e.g., it cannot be reduced to a logarithmic or lower dependence) even for bounded distributions, yielding a tight communication-complexity bound (see \cref{cc-mult} in \cref{extensions}). Alternatively to fixing $n$ and analyzing the menu-size and communication complexity as functions of $\varepsilon$ as we do, one may fix $\varepsilon$ and analyze these quantities as functions of~$n$. For any fixed~$\varepsilon>0$, the upper bound of \citet{bgn17} is exponential in~$n$; therefore, another question left open by that paper is whether this exponential dependence may be avoided. (In terms of communication complexity, this question asks whether for every fixed $\varepsilon$, the communication complexity can be logarithmic or even polylogarithmic in~$n$.) Some progress on this question has been made already by \citet{bgn17}, who show that on the one hand, a polynomial dependence on~$n$ suffices for some values of~$\varepsilon$ (namely, $\varepsilon\ge\nicefrac{5}{6}$), and that on the other hand, an exponential dependence on $n$ is required when coupling $\varepsilon$ with the number of goods $n$ by setting $\varepsilon\eqdef\nicefrac{1}{n}$;\footnote{This result also shows that the exponential dependence on $n$ in the upper bounds of \citet{hn13} and \citet{dhn14} for \emph{additive} up-to-$\varepsilon$ approximation (see \cref{extensions}) is required even for bounded product distributions. However, one may claim that the ``more interesting'' question when keeping~$\varepsilon$ fixed and letting $n$ vary is that of a multiplicative approximation, as asking for at most an $\varepsilon$ additive loss while increasing the total value in the market (by adding more and more items) is quite a harsh requirement.} however, as noted above, they leave the general case of arbitrary fixed~$\varepsilon>0$ (uncoupled from~$n$, e.g., $\varepsilon=1\%$) as their main remaining open question.
While the current state-of-the-art literature seems to be a long way from identifying very-high-dimensional optimal mechanisms, and especially from identifying their duals (indeed, it took quite some impressive effort for \citet{gk14} to identify the optimal mechanism for $6$ goods whose valuations are i.i.d.\ uniform in $[0,1]$), one may hope that with time, it may be possible to do so. It seems plausible that if one could generate high-dimensional optimal mechanisms (and corresponding duals) for which the high-dimensional analogue of the curve that we denote by $S$ in \cref{sec:lower} has large-enough measure (while maintaining a small-enough radius of curvature, etc.), then an argument similar to the one that we give in the proof of \cref{lower} may be used to show that an exponential dependence on~$n$ in the above bound is indeed required for sufficiently small, yet fixed, $\varepsilon$, and thereby resolve the above open question. Whether one can generate such mechanisms with large-enough high-dimensional analogues of $S$, however, remains to be seen.

\section*{Acknowledgments}
Yannai Gonczarowski is supported by the Adams Fellowship Program of the Israel Academy of Sciences and Humanities.
His work is supported by ISF grant 1435/14 administered by the Israeli Academy of Sciences and by Israel-USA Bi-national Science Foundation (BSF) grant number 2014389.
This project has received funding from the European Research Council (ERC) under the European Union's
Horizon 2020 research and innovation programme (grant agreement No 740282).
I thank Costis Daskalakis, Kira Goldner, Sergiu Hart, Ian Kash, Noam Nisan, Christos Tzamos, and an anonymous referee for helpful feedback. I thank Costis Daskalakis, Alan Deckelbaum, and Christos Tzamos for permission to base \cref{gamma} on a \lcnamecref{gamma} from their paper.

\bibliographystyle{abbrvnat}
\bibliography{menu-transport}

\end{document}